\documentclass[letterpaper, 10 pt, conference]{ieeeconf} 
\IEEEoverridecommandlockouts
\usepackage{cite}

\usepackage{amsmath,amssymb,amsfonts}
\usepackage{graphicx}
\usepackage{textcomp}
\usepackage{xcolor}
\usepackage{amsmath, amssymb, amsfonts} 
\usepackage{algorithm}
\usepackage{algpseudocode}
\usepackage{tikz}
\usetikzlibrary{positioning, arrows, fit}
\usetikzlibrary{calc}
\usetikzlibrary{decorations.markings,patterns,decorations.pathmorphing}
\usetikzlibrary{arrows, shapes}
\tikzstyle{block} = [draw, thick, node distance=0.5cm, minimum width=1cm, inner sep=6pt]
\tikzstyle{sum} = [draw, thick, circle, node distance=1cm, inner sep=3.5pt, path picture={\node at (path picture bounding box.center) [draw, anchor = center] {$+$};}]

\let\labelindent\relax
\usepackage{enumitem}

\usepackage{booktabs} 
\usepackage{multirow}
\usepackage{url}
\usepackage{amsthm}

\usepackage{mathtools}
\usepackage{dsfont}
\usepackage{comment} 

\usepackage{arydshln}

\newtheorem{lemma}{Lemma}
\newtheorem{theorem}{Theorem}
\newtheorem{proposition}{Proposition}

\newtheorem{remark}{Remark}

\newcommand{\nxfull}{n_{\mathrm{f}}}
\newcommand{\nxreduced}{n_{\mathrm{r}}}
\newcommand{\ninput}{n_{\mathrm{u}}}
\newcommand{\npsi}{n_{\psi}}
\newcommand{\nz}{n_{\mathrm{z}}}
\newcommand{\nw}{n_{\mathrm{w}}}
\usepackage{etoolbox}
\newbool{arxiv}
\setbool{arxiv}{true}
\usepackage{setspace}
\title{\LARGE \bf
Robust reduced-order model predictive control\\ using
peak-to-peak analysis of filtered signals}

\author{Johannes Köhler$^{1,2}$, Carlo Scholz$^1$, Melanie Zeilinger$^1$%
\thanks{$^1$Institute for Dynamic Systems and Control,
ETH Zürich, Switzerland.}
\thanks{$^2$Department of Mechanical Engineering, Imperial College London, London, UK. Email: \texttt{j.kohler@imperial.ac.uk}}
}

\begin{document}

\maketitle
\thispagestyle{empty}
\pagestyle{empty}
\begin{abstract}
We address the design of a model predictive control (MPC) scheme for large-scale linear systems using reduced-order models (ROMs). 
Our approach uses a ROM, leverages tools from robust control, and integrates them into an MPC framework to achieve computational tractability with robust constraint satisfaction. 
Our key contribution is a method to obtain guaranteed bounds on the predicted outputs of the full-order system by predicting a (scalar) error-bounding system alongside the ROM. 
This bound is then used to formulate a robust ROM-based MPC that guarantees constraint satisfaction and robust performance.
Our method is developed step-by-step by (i) analysing the error, (ii) bounding the peak-to-peak gain, an (iii) using filtered signals. 
We demonstrate our method on a 100-dimensional mass-spring-damper system, achieving over four orders of magnitude reduction in conservatism relative to existing approaches.
\end{abstract}
\ifbool{arxiv}{\begin{keywords}
Predictive control for linear systems; 
Reduced order modeling;
Robust control
\end{keywords}}{}
\section{Introduction}
\label{sec:intro}
Model predictive control (MPC) is an advanced control method that leverages online optimization and model-based prediction to satisfy safety-critical constraints~\cite{rawlings2017model}. 
MPC generates feedback by solving a finite-horizon optimal control problem and thus computational demand is an important factor for industrial applications. 
In the last decades, advances in numerical algorithms (cf., e.g.,~\cite{kouzoupis2018recent}) have expanded the applicability of MPC from slow systems in process control~\cite{qin2003survey} to faster systems, such as robots~\cite{katayama2023model} or power electronics~\cite{rodriguez2021latest}. 
However, application to large-scale systems, such as power grids, deformable structures, and partial differential equations, remains computationally challenging. 
Lower-dimensional \emph{reduced-order} models (ROMs)~\cite{antoulas2005overview} are a promising approach to alleviate this issue. 
Applicability of ROM-based MPC has been empirically demonstrated, for example, in process control~\cite{huisman2003identification,marquez2013model}, soft robotics~\cite{tonkens2021soft}, and humanoid locomotion~\cite{ghansah2025hierarchical}. 
However, ROMs also introduce significant prediction errors, which can be problematic in safety-critical applications. 
In this paper, we propose a systematic framework to rigorously account for prediction errors induced by ROMs, thereby guaranteeing the \emph{robust} satisfaction of safety-critical constraints. 

\subsection{Related work}

\subsubsection*{Error bounds for reduced-order models}
Classical error bounds for ROMs involve, e.g., $\mathcal{H}_\infty$ or $\mathcal{H}_2$ bounds of balanced truncation~\cite{antoulas2005overview}. Similarly, $(\gamma,\delta)$-similarity characterizations have recently been utilized for ROMs~\cite{bajaj2025online}.
While such bounds are well-suited for analysing \emph{stability}, they do \emph{not} provide point-wise bounds required for robust predictions. 
Point-wise bounds for nonlinear ROMs are, e.g., derived in~\cite{wirtz2014posteriori}. 
However, these bounds may grow exponentially over the prediction horizon since they rely on the (one-sided) Lipschitz constant.
\subsubsection*{Robust MPC with reduced-order models}
A key challenge in designing a predictive controller with ROMs is ensuring satisfaction of safety-critical constraints despite the prediction error. 
A common approach in \emph{robust MPC} is to assume a uniform bound on the model mismatch, which enables simple bounds on the prediction error~\cite{kouvaritakis2016model}. 
In a similar spirit, compact state and input constraints can be utilized to derive uniform bounds on the prediction error with ROMs~\cite{sopasakis2013constrained,kogel2015robust,lorenzetti2022linear,wang2022tube}. 
However, the prediction error of ROMs depends primarily on the operation of the system and neglecting this dependency through uniform error bounds can be overly conservative. 
To address this limitation, input-dependent prediction errors for linear and nonlinear ROMs are derived in~\cite{loehning2014model}
and \cite{alora2023robust}, respectively. 
These approaches yield robust MPC schemes that can be efficiently implemented by predicting a scalar error-bounding system alongside the ROM. 
However, these error bounding systems rely on hyper-parameters which need to be carefully tuned. 
Our work also derives a stable error-bounding system, but leverages robust control theory to enable a systematic and less conservative design. 

\subsubsection*{Robust control for dynamic uncertainties}
Reduced-order modelling results in unmodelled dynamics, which correspond to \emph{dynamic uncertainties} in robust control~\cite[Chap.~7]{zhou1998essentials}. 
Robust MPC designs for dynamic uncertainty are, e.g., presented in~\cite{falugi2014getting,lovaas2008robust}. 
Similar to the ROM-based MPC schemes~\cite{sopasakis2013constrained,kogel2015robust,lorenzetti2022linear,wang2022tube}, these works use uniform bounds on the prediction error based on compact constraints. 
More recently, input-dependent prediction error bounds with dynamic uncertainties have been utilized in~~\cite{schwenkel2022model,schwenkel2025output}. 
These approaches leverage (time-domain) \emph{integral quadratic constraints (IQCs)}~\cite{scherer2022dissipativity}, a general formulation to model and characterize dynamic uncertainty. 
IQCs use dynamic filters to characterize the input-output behavior of the dynamic uncertainty, naturally encompassing classical robust analysis tools, such as small-gain or dissipativity.
Input-dependent error bounds are then obtained by analysing the worst-case \emph{peak-to-peak gain}~\cite{schwenkel2023robust,schwenkel2025multi}. Our approach adapts some of these tools for ROMs.

\subsection{Contribution}
We consider the design of a robust MPC scheme for large-scale systems using reduced-order models (ROMs).  
Our results apply to linear continuous-time systems subject to hard input and output constraints. 
Our main contribution is a scalar error-bounding system that robustly quantifies the deviation between the prediction of the ROM and the full-order system depending on the predicted trajectory. 
We gradually build up this error-bounding system by analysing the peak-to-peak gain of the error dynamics and using dynamic filters to introduce a suitable weighting. 
We provide a ROM-based MPC formulation that guarantees constraint satisfaction and a performance bound for the full-order system. 
We demonstrate the effectiveness of our approach through a numerical example involving a $100$-dimensional mass-spring-damper system. 
We provide a qualitative and quantitative comparison to existing bounds for ROMs~\cite{loehning2014model,lorenzetti2022linear}, showing improvements of over four orders of magnitude. 
\ifbool{arxiv}{As a result of independent interest, Appendix~\ref{app:IQC} adapts the results in~\cite{schwenkel2022model,schwenkel2023robust} to
 provide robust prediction bounds for uncertain linear continuous-time systems based on IQCs and discusses the relation to the proposed methodology.}{}

\subsubsection*{Notation}
For a vector $x\in\mathbb{R}^n$, we denote the Euclidean norm by $\|x\|=\sqrt{x^\top x}$ and the infinity norm by $\|x\|_\infty$. 
For a positive semi-definite matrix $P\in\mathbb{R}^{n\times n}$, we denote $\|x\|_P^2=x^\top P x$. 
We denote the identity matrix by $I_n\in\mathbb{R}^{n\times n}$, the zero matrix by $0_{n,m}\in\mathbb{R}^{n\times m}$, and the vector of ones by $\mathbf{1}_n\in\mathbb{R}^n$, where we also drop the subindex if the dimensions are clear from the context. 
For a matrix $A=A^\top\in\mathbb{R}^{n\times n}$, $A\succ 0$ ($\succeq 0$) denotes that the matrix is positive definite (positive semi-definite). 
Symmetric forms $A^\top P A$ are abbreviated by $[\star]^\top P A$. 
For a signal $x(t)\in\mathbb{R}^n$, $t\in\mathbb{R}$, we denote the time-derivative at time $t$ by $\dot{x}(t)\in\mathbb{R}^n$. 
The peak-norm of a signal 
$x(t)\in\mathbb{R}^n$ is denoted by $\|x\|_{\mathrm{peak}}=\sup_{t\geq 
 0}\|x(t)\|$. 

\section{Problem setup}
\label{sec:setup}
We consider a linear continuous-time system
\begin{align}
\label{eq:sys}
\Sigma: 
\begin{cases}\dot{x}(t)=Ax(t)+Bu(t)+Ew(t), ~x(0)=x_0,\\
z(t)=Cx(t),
\end{cases}
\end{align}
with state $x(t)\in\mathbb{R}^{\nxfull}$, control input $u(t)\in\mathbb{R}^{\ninput}$, disturbances $w(t)\in\mathbb{R}^{\nw}$, controlled output $z(t)\in\mathbb{R}^{\nz}$, and time $t\geq 0$. 
The disturbances satisfy $w(t)\in\mathcal{W}=\{w\in\mathbb{R}^{\nw}|~\|w\|\leq \bar{w}\}$, $\forall t\geq 0$. 
Our goal is to track a desired reference trajectory $z_{\mathrm{ref}}(t)\in\mathbb{R}^{\nz}$ while robustly satisfying input and output constraints
\begin{align}
\label{eq:constraints}
z(t)\in\mathcal{Z},~u(t)\in\mathcal{U},~\forall t\geq 0,~\forall w(t)\in\mathcal{W}.
\end{align}
The constraints are given by
\begin{align}
\label{eq:constraint_con}
\mathcal{Z}=\{z\in\mathbb{R}^{\nz}|~g_{j}(z)\leq 0,~j=0,\dots,n_{\mathrm{c}}\} 
\end{align}
with $g_j(z)$ Lipschitz continuous with constants $L_{\mathrm{g},j}\geq 0$. 
These conditions include polytopic and conic constraints $\mathcal{Z}$. 

\subsubsection*{Simplifying conditions}
In this paper, we assume that: 
The system~\eqref{eq:sys} is open-loop stable, i.e., $A$ is Hurwitz, and we only optimize open-loop inputs.\footnote{This can be naturally relaxed by using an affine parametrization $u(t)=Kx(t)+v(t)$ with an offline optimized feedback $K$, see Remark~\ref{rk:feedback}.}
The model equations $(A,B,E,C)$ and the initial condition $x_0\in\mathbb{R}^{\nxfull}$ are perfectly known.
We only consider operation over a finite-horizon $T\geq 0$. 
The control input $u(t)$ and the disturbances $w(t)$ are (piece-wise) continuous.

\subsubsection*{Full-order optimal control problem}
The problem reduces to the following (standard) finite-horizon constrained robust optimal control problem
\begin{subequations}
\label{eq:OCP}
\begin{align}
\label{eq:OCP_cost}
\min_{u(\cdot)\in\mathcal{U}}\max_{w(\cdot)\in\mathcal{W}}&\int_{t=0}^T \|z(t)-z_{\mathrm{ref}}(t)\|_\infty+R\cdot\|u(t)\|_\infty~\mathrm{d}t\\
\text{s.t. } &\eqref{eq:sys},~\eqref{eq:constraints},~\forall w(t)\in\mathcal{W},~t\in[0,T],
\end{align}
\end{subequations}
where $R>0$ is the input weighting. The setup can be naturally generalized to any Lipschitz continuous cost. 
The key challenge is that the dimension of the state $x(t)\in\mathbb{R}^{\nxfull}$ is very large, i.e., $\nxfull\gg 1$. This makes Problem~\eqref{eq:OCP} computationally intractable.

\subsubsection*{Proposed approach \& Outline}
We determine a feasible, but possibly suboptimal, input trajectory $u(t)$ for Problem~\eqref{eq:OCP} by leveraging ROMs and robust control tools.
We first introduce the ROM and characterize the error between the full-order system and the ROM (Sec.~\ref{sec:ROM}). 
Then, we derive an error-bounding system for the ROM by analysing the peak-to-peak gain (Sec.~\ref{sec:peak}). 
We extend this methodology to leverage the peak-to-peak gain based on filtered signals (Sec.~\ref{sec:filter}). 
Based on this, we introduce a ROM-based MPC formulation that ensures robust constraint satisfaction for the full-order system (Sec.~\ref{sec:mpc}). 
We qualitatively compare this approach to existing methods (Sec.~\ref{sec:discussion}). 
Lastly, we provide quantitative comparisons with a numerical example involving a $100$-dimensional mass-spring-damper system (Sec.~\ref{sec:num}).

\section{Reduced-order models \\ and dynamic uncertainty}
\label{sec:ROM}
We construct a ROM with state $x_{\mathrm{r}}(t)\in\mathbb{R}^{\nxreduced}$, where $\nxreduced\ll \nxfull$. 
We use a standard Petrov-Galerkin projection with matrices $W,V\in\mathbb{R}^{\nxfull\times\nxreduced}$ satisfying $W^\top V=I_{\nxreduced}$. 
The nominal ROM is given by
\begin{align}
\label{eq:sys_ROM}
\Sigma_{\mathrm{r}}
\begin{cases}
&\dot{x}_{\mathrm{r}}(t)=A_{\mathrm{r}} x_{\mathrm{r}}(t)+B_{\mathrm{r}} u(t),~x_{\mathrm{r}}(0)=x_{\mathrm{r},0},\\
&z_{\mathrm{r}}(t)=C_{\mathrm{r}} x_{\mathrm{r}}(t),
\end{cases}
\end{align}
with $A_{\mathrm{r}}=W^\top A V$, $B_{\mathrm{r}}=W^\top B$, $C_{\mathrm{r}}=CV$, and initial condition $x_{\mathrm{r},0}=W^\top x_0$. 
To study the difference between the two models we define the full-order (lifted) error: 
\begin{align}
\label{eq:error_def} 
e(t):=&x(t)-Vx_{\mathrm{r}}(t)\in\mathbb{R}^{\nxfull},~t\geq 0.
\end{align}
\begin{proposition}[Error dynamics]
\label{prop:error_dyn}
The full-order system $\Sigma$~\eqref{eq:sys} is equivalent to the interconnection between the reduced-order system $\Sigma_{\mathrm{r}}$~\eqref{eq:sys_ROM} and the following error dynamics
\begin{align}
\label{eq:error_dyn}
\Delta: \begin{cases}
\dot{e}(t)=Ae(t)+d(t)+Ew(t),~e(0)=e_0,\\
d(t)=(I_{\nxfull}-VW^\top)(AVx_{\mathrm{r}}(t)+Bu(t)),\\
z_{\mathrm{e}}(t)=Ce(t),
\end{cases}
\end{align}
with $e_0=(I_{\nxfull}-VW^\top)x_0$ and $z(t)=z_{\mathrm{r}}(t)+z_{\mathrm{e}}(t)$.
\end{proposition}
\begin{proof}
Given~\eqref{eq:sys} and \eqref{eq:sys_ROM}, the error~\eqref{eq:error_def} satisfies
\begin{align*}
&\dot{e}(t)=\dot{x}(t)-V\dot{x}_{\mathrm{r}}(t)\\
=&Ax(t)+Bu(t)+Ew(t)-V(A_{\mathrm{r}}x_{\mathrm{r}}(t)+B_{\mathrm{r}}u(t))\\
=&Ae(t)+(AV-VA_{\mathrm{r}})x_{\mathrm{r}}(t)+(B-VB_{\mathrm{r}})u(t) +Ew(t)\\
=&Ae(t)+(I_{\nxfull}-VW^\top )(A Vx_{\mathrm{r}}(t)+Bu(t))+Ew(t),
\end{align*}
with $e_0=x_0-Vx_{\mathrm{r},0}=(I_{\nxfull}-VW^\top)x_0$.
\end{proof}
The error $e(t)$ is driven by the model-mismatch $d(t)$, which depends also on the reduced-order state $x_{\mathrm{r}}(t)$ and the control input $u(t)$. 
The factor $I_{\nxfull}-VW^\top$ corresponds to the residual of oblique projection, which vanishes as $\nxreduced\rightarrow \nxfull$. 
With Proposition~\ref{prop:error_dyn}, we can characterize the full-order system $\Sigma$ as an interconnection of the reduced-order system $\Sigma_{\mathrm{r}}$ and the unmodelled dynamics $\Delta$. 

\begin{remark}[Choice of decomposition]
The decomposition of the full-order system in a ROM $\Sigma_{\mathrm{r}}$ and unmodelled dynamics $\Delta$ is non-unique. 
We considered the \emph{nominal} ROM~\eqref{eq:sys_ROM}, which reduces robust predictions to deriving a bound on the output of $\Delta$ driven by external inputs $d(t)$ and $w(t)$. 
The problem could be equivalently represented by a feedback interconnection of the $n_{\mathrm{r}}$ reduced-order states $W^\top x(t)\neq x_{\mathrm{r}}(t)$ and the remaining $n_{\mathrm{f}}-n_{\mathrm{r}}$ unmodelled states, both perturbed by external signals~\cite[Lemma~2]{alora2023robust}. This equivalent decomposition can be more natural for receding-horizon implementation, which is, however, not the focus of this work. 
\end{remark}

\section{Peak-to-peak analysis \& robust prediction}
\label{sec:peak}
In this section, we derive a robust prediction error bound for the ROM~\eqref{eq:sys_ROM}, which depends on the applied input $u(t)$ and the predicted state $x_{\mathrm{r}}(t)$. 
Specifically, we write the error dynamics~\eqref{eq:error_dyn} equivalently as:
\begin{align}
\label{eq:error_dyn_input}
\dot{e}(t)=&Ae(t)+B_{\mathrm{e}}r(t),~e(0)=e_0,\\
r(t)=&[w(t)^\top,x_{\mathrm{r}}(t)^\top,u(t)^\top]^\top\in\mathbb{R}^{\nw+\nxreduced+\ninput},\nonumber\\
z_{\mathrm{e}}(t)=&Ce(t),\nonumber\\
B_{\mathrm{e}}=&\begin{bmatrix}
E&
(I_{\nxfull}-VW^\top)\begin{bmatrix}
AV&B 
\end{bmatrix}
\end{bmatrix}, \nonumber
\end{align}
with the lumped input $r(t)$. 
Determining the magnitude of the output error $z_{\mathrm{e}}(t)$ depending on the magnitude of the input $r(t)$ can be cast as analysing the peak-to-peak norm, i.e., determining a constant $\gamma>0$, such that 
\begin{align}
\label{eq:peak_norm} 
\|z_{\mathrm{e}}\|^2_{\mathrm{peak}}\leq \gamma^2\|r\|_{\mathrm{peak}}^2+c\|e_0\|^2, c\geq 0,
\end{align}
 for any input signal $r(t)$.
As we demonstrate later in Proposition~\ref{prop:peak_gain}, this problem can be addressed with the following optimization problem, adapted from~\cite[Prop.~3.16]{scherer2000linear}: 
\begin{subequations} 
\label{eq:LMI_peak}
\begin{align}
\label{eq:LMI_peak_cost}
\min_{P,\lambda,\gamma}&~~\gamma\\
\label{eq:LMI_peak_Lyap}
\text{s.t. } &
\left[ \vphantom{
\begin{bmatrix}
x\\x\\x
\end{bmatrix} } 
 \star \right]^{\top}
\begin{bmatrix}\lambda P &P&0\\P&0&0\\0&0&-\gamma I\end{bmatrix}
\begin{bmatrix}
I_{\nxfull}&0\\
A&B_{\mathrm{e}}\\
0&I
\end{bmatrix}\preceq 0, \\
\label{eq:LMI_peak_output}
&\begin{bmatrix}
\gamma I_{\nz}&C\\
C^\top & \lambda P
\end{bmatrix}\succeq 0,~P\succeq 0.
\end{align}
\end{subequations}
For a fixed constant $\lambda>0$, Problem~\eqref{eq:LMI_peak} is a semi-definite program, which can be efficiently solved. The exponential decay rate $\lambda>0$ is determined using an additional linesearch.
Feasibility of Problem~\eqref{eq:LMI_peak} follows from $A$ Hurwitz, with a suitable choice of $\lambda.$ 
This LMI yields the Lyapunov function $V(t)=\|e(t)\|_P^2$, which we use to sequentially determine a bound on $z_{\mathrm{e}}(t)$ in the following proposition. 
\begin{proposition}
\label{prop:peak_gain} 
The following scalar error-bounding system
\begin{align}
\label{eq:error_bounding_system}
\dot{\delta}(t)=-\lambda \delta(t)+\gamma \|r(t)\|^2,~\delta(0)=V(0)
\end{align}
satisfies
\begin{align}
\label{eq:peak_error_bound}
\dfrac{1}{\lambda \gamma}\|z_{\mathrm{e}}(t)\|^2\leq V(t)\leq \delta(t),~\forall t\geq 0. \end{align} 
Furthermore, the error dynamics~\eqref{eq:error_dyn_input} have a peak-to-peak norm smaller or equal to $\gamma\geq 0$, i.e., \eqref{eq:peak_norm} holds.
\end{proposition}
\begin{proof}
Multiplying~\eqref{eq:LMI_peak_Lyap} from left and right by $[e(t)^\top,~ r(t)^\top]$ and its transposed yields
\begin{align}
\label{eq:peak_lyap_decrease}
 \lambda V(t)+\dot{V}(t) - \gamma \|r(t)\|^2\leq 0.
\end{align}
The LMI~\eqref{eq:LMI_peak_output} can be equivalently reformulated using a Schur complement (cf.~\cite{boyd1994linear}) as:
\begin{align*}
-\lambda P +C^\top (\gamma I_{\nz})^{-1}C
\preceq 0.
\end{align*}
Multiply with $e(t)$ and its transposed from left and right yields
\begin{align*}
 -\lambda V(t)+\dfrac{1}{\gamma}\|z_{\mathrm{e}}(t)\|^2\leq 0,
\end{align*}
i.e., the first inequality in~\eqref{eq:peak_error_bound} holds. 
We show that $V(t)\leq \delta(t)$ $\forall t\geq 0$ by using the comparison lemma. 
For $t=0$, we have 
$V(0)=\delta(0)$. 
To show $V(t)\leq \delta(t)$ $\forall t\geq 0$, it suffices if $V(t)=\delta(t)$ implies $\dot{V}(t)\leq \dot{\delta}(t)$. This follows from~\eqref{eq:error_bounding_system} and \eqref{eq:peak_lyap_decrease}.
The peak-to-peak gain~\eqref{eq:peak_norm} holds by 
\begin{align*}
\dfrac{1}{\lambda \gamma}\|z_{\mathrm{e}}(t)\|^2\stackrel{\eqref{eq:peak_error_bound}}{\leq} \delta(t)\stackrel{\eqref{eq:error_bounding_system}}{\leq} \exp(-\lambda t)\|e_0\|_P^2+\dfrac{\gamma}{\lambda}\|r\|^2_{\mathrm{peak}}
\end{align*}
and setting $c>0$ to the maximum eigenvalue of $P$.
\end{proof}
The resulting error-bounding system is comparable to~\cite[Thm.~1]{loehning2014model}. 
The main difference is that we have a systematic way to determine the constants $\lambda,\gamma$, by minimizing the peak-to-peak gain in~\eqref{eq:LMI_peak_cost}.

\section{Robust predictions using peak-to-peak analysis of filtered signals}
\label{sec:filter}
In this section, we extend the robust prediction based on the peak-to-peak analysis from Proposition~\ref{prop:peak_gain} by introducing an additional dynamic filter.

\subsubsection*{Motivation for considering filtered signals}
The error-bounding system~\eqref{eq:error_bounding_system} determines a bound on the prediction error $z_{\mathrm{e}}(t)$ based on the magnitude of the lumped input $r(t)=[w(t)^\top,~x_{\mathrm{r}}(t)^\top,~u(t)^\top]^\top\in\mathbb{R}^{\nw+\nxreduced+\ninput}$. 
Compared to existing uniform bounds for ROMs~\cite{sopasakis2013constrained,kogel2015robust,lorenzetti2022linear,wang2022tube}, this captures the fact that the prediction error decreases if the magnitude of the excitation $r(t)$ decreases.
However, even with the same magnitude of $r(t)$, the prediction error can vary drastically. 
For example, ROMs are often chosen based on the slowest modes and thus the magnitude of the prediction error also depends on the frequency components of the excitation $r(t)$. 
This effect cannot be captured by the approach in Section~\ref{sec:peak}, which determines a bound proportional to the peak norm of $r(t)$. 
In the following, we generalise this by instead considering a filtered signal $r_\psi(t)$. By suitably choosing the filter, this can better capture the dependence of the prediction error on the excitation, see Sections~\ref{sec:discussion} and \ref{sec:num} for an application to a numerical example. 
\subsubsection*{Dynamic filter}
We consider a general linear dynamic filter 
\begin{align}
\label{eq:filter}
\Psi
\begin{cases}
&~\dot{\psi}(t)=A_\psi \psi(t)+B_\psi r(t),~\psi(0)=0,\\
&r_{\psi}(t)=C_{\psi}\psi(t) +D_{\psi}r(t),
\end{cases}
\end{align}
which generates a filtered output $r_\psi(t)\in\mathbb{R}^{n_{\psi,\mathrm{y}}}$ based on the lumped input $r(t)$.
The filter is a design choice.  
We choose a stable filter, i.e., $A_\psi\in\mathbb{R}^{n_\psi\times n_\psi}$ is Hurwitz. 
Furthermore, $B_\psi$ is chosen such that the disturbance $w(t)$ does not affect the filter state $\psi(t)$, but is simply passed through to the output $r_\psi(t)$ with the direct-feed-through $D_\psi$, i.e., 
\begin{align}
\label{eq:filter_structure}
[I_{\nw},0]r_\psi(t)=w(t),~\begin{bmatrix}B_\psi\\ \begin{bmatrix}0_{\nw}& I_{\nxreduced+\ninput}\end{bmatrix}D_\Psi\end{bmatrix}
\begin{bmatrix}w(t)\\0\end{bmatrix}=0.
\end{align}
This choice does not affect the following analysis, but is crucial for implementability of the MPC formulation in Section~\ref{sec:mpc} later.

We write the augmented system consisting of the error dynamics~\eqref{eq:error_dyn_input} and the filter~\eqref{eq:filter} as
\begin{align}
\label{eq:dyn_aug_error_filter}
\dot{\chi}(t)=&A_\chi\chi(t)+B_\chi r(t),~\chi(0)=[e_0^\top,~0]^\top,\\
r_\psi(t)=&C_\chi\chi(t)+D_\chi r(t),\nonumber\\
z_{\mathrm{e}}(t)=&C_{\mathrm{z}}\chi(t),\nonumber\\
&A_\chi=\begin{bmatrix}A&\\&A_\psi\end{bmatrix}, ~
B_\chi=\begin{bmatrix}B\\B_\psi \end{bmatrix},\nonumber\\
&C_\chi=[0,C_\psi],~D_\chi=D_\psi,~C_{\mathrm{z}}=[C,~0].\nonumber
\end{align}
We are interested in a peak-to-peak gain from the filtered input $r_\psi(t)$ to the output error $z_{\mathrm{e}}(t)$. 
The corresponding reformulation is given by the optimization problem below:
\begin{subequations} 
\label{eq:LMI_peak_filter}
\begin{align}
\label{eq:LMI_peak_filter_cost}
&\min_{P_\chi,\lambda_\chi,\gamma_\chi}~~\gamma_\chi\\
\label{eq:LMI_peak_filter_Lyap}
\text{s.t. } &
\left[ \vphantom{
\begin{bmatrix}
x\\x\\x
\end{bmatrix} } 
 \star \right]^{\top}
\begin{bmatrix}\lambda_\chi P_\chi &P_\chi&\\P_\chi&0&\\&&-\gamma_\chi I_{n_{\psi,\mathrm{y}}}\end{bmatrix}
\begin{bmatrix}
I_{\nxfull+n_\psi}&0\\
A_\chi&B_\chi\\
C_\chi&D_\chi
\end{bmatrix}\preceq 0, \\
\label{eq:LMI_peak_filter_output}
&\begin{bmatrix}
\gamma_\chi I_{\nz}&C_{\mathrm{z}}\\
C_{\mathrm{z}}^\top & \lambda_\chi P_\chi
\end{bmatrix}\succeq 0,~P_\chi\succeq 0.
\end{align}
\end{subequations}
Since $r_\psi(t)$ is an output of the system~\eqref{eq:dyn_aug_error_filter}, we technically investigate an output-to-output gain~\cite{arnstrom2025efficiently}.\footnote{%
For invertible filters $\Psi$, it may be possible to rewrite the conditions equivalently as a standard peak-to-peak gain from an input to an output. However, this might complicate implementation of the proposed MPC formulation in Section~\ref{sec:mpc}.} 
In the following, we suppose that Problem~\eqref{eq:LMI_peak_filter} admits a feasible solution and define $V_\chi(t)=\chi(t)^\top P_\chi \chi(t)$. 
We consider the following error-bounding system based on the filtered output $r_\psi(t)$:
\begin{align}
\label{eq:filter_error_bounding_system}
\dot{\delta}_\chi(t)=-\lambda_\chi \delta_\chi(t)+\gamma_\chi \|r_\psi(t)\|^2,~\delta_\chi(0)=V_\chi(0).
\end{align}
\begin{proposition}
\label{prop:filter_peak_gain}
It holds that 
\begin{align}
\label{eq:peak_filter_error_bound}
\dfrac{1}{\lambda_\chi \gamma_\chi}\|z_{\mathrm{e}}(t)\|^2\leq V_\chi(t)\leq \delta_\chi(t),~\forall t\geq 0, \end{align} 
and 
\begin{align}
\label{eq:filter_peak_norm} 
\|z_{\mathrm{e}}\|_{\mathrm{peak}}^2\leq \gamma_\chi^2\|r_\psi\|_{\mathrm{peak}}^2+c_\chi\|e_0\|^2,~c_\chi\geq 0.
\end{align}
\end{proposition}
\begin{proof}
The proof is analogous to Proposition~\ref{prop:peak_gain}, using the augmented dynamics~\eqref{eq:dyn_aug_error_filter}.
\end{proof}
The error-bounding system from Section~\ref{sec:peak} is contained as a special case by simply selecting $r_\psi(t)=r(t)$. 
As we see later in the numerical experiments (Sec.~\ref{sec:num}), choosing an appropriate filter $\Psi$ is vital to avoid overly conservative error bounds, however, a systematic optimization of the filter is left for future work. 

\section{Robust predictive control with ROMs}
\label{sec:mpc}
In the following, we derive the proposed MPC formulation, which leverages the ROM~\eqref{eq:sys_ROM} and the error-bounding system~\eqref{eq:filter_error_bounding_system}.

\subsubsection*{Robust prediction}
Simulating the scalar-error bounding system~\eqref{eq:filter_error_bounding_system} requires the filter output $r_\psi(t)$, which also depends on the unknown disturbance $w(t)$~\eqref{eq:filter}. 
Recall that $\|w(t)\|\leq \bar{w}$ with a know constant $\bar{w}\geq 0$. 
We leverage the structure of the filter~\eqref{eq:filter},\eqref{eq:filter_structure} to define an implementable robust prediction:
\begin{subequations} 
\label{eq:full_prediction}
\begin{align}
\dot{x}_{\mathrm{r}}(t)=&A_{\mathrm{r}}x_{\mathrm{r}}(t)+B_{\mathrm{r}}u(t),~x_{\mathrm{r}}(0)=W^\top x_0,\\
\label{eq:full_prediction_filter_dyn}
\dot{\bar{\psi}}(t)=&A_\psi \bar{\psi}(t) +B_\Psi \bar{r}(t),~\bar{\psi}(0)=0,\\
\label{eq:full_prediction_delta_dyn}
\dot{\bar{\delta}}_{\chi}(t)=&-\lambda_\chi\bar{\delta}_{\chi}(t)+\gamma_{\chi}(\|\bar{r}_\psi(t)\|^2+\bar{w}^2),\\
\bar{\delta}_\chi(0)=&\|[(I_{\nxfull}-VW^\top)x_0;~0]\|_{P_\chi}^2,\nonumber\\
\label{eq:full_prediction_delta_z}
\delta_{\mathrm{z}}(t)=&\sqrt{\gamma_\chi\lambda_\chi \bar{\delta}_{\chi}(t)},\\
\bar{r}_\psi(t)=&C_\psi \bar{\psi}(t)+D_\psi \bar{r}(t),\\
\bar{r}(t)=&\begin{bmatrix}
 0 & x_{\mathrm{r}}(t)^\top & u(t)^\top 
\end{bmatrix}^\top,\\
z_{\mathrm{r}}(t)=&C_{\mathrm{r}}x_{\mathrm{r}}(t).
\end{align}
\end{subequations}
Equations~\eqref{eq:full_prediction} depend only on known dynamics and can be efficiently implemented.
\begin{lemma}
\label{lemma:full_error_bound}
The full-order system~\eqref{eq:sys} and the dynamics~\eqref{eq:full_prediction} satisfy
\begin{align}
\label{eq:robust_prediction}
\|z_{\mathrm{r}}(t)-z(t)\|\leq \delta_{\mathrm{z}}(t),~\forall t\geq 0.
\end{align}
\end{lemma}
\begin{proof}
The results follows largely from Proposition~\ref{prop:filter_peak_gain}. 
First, note that the structure of the filter~\eqref{eq:filter}--\eqref{eq:filter_structure} and the "nominal filter" (which neglects the disturbance $w(t)$)~\eqref{eq:full_prediction_filter_dyn}  satisfy 
$r_\psi(t)=\bar{\psi}(t)+[I,~0]^\top w(t)$. Furthermore, $\|w(t)\|\leq\bar{w}$, \eqref{eq:filter_error_bounding_system}, \eqref{eq:full_prediction_delta_dyn} and monotonicity of the scalar error-bounding system ensure
$\bar{\delta}_{\chi}(t)\geq \delta_\chi(t)$, $\forall t\geq 0$.
Lastly, \eqref{eq:robust_prediction} follows from \eqref{eq:peak_filter_error_bound} and \eqref{eq:full_prediction_delta_z}. 
\end{proof}
 
\subsubsection*{Robust constraint satisfaction}
In addition, to the robust bounds on the prediction (Lemma~\ref{lemma:full_error_bound}), we need an efficient formulation to enforce the output constraints~\eqref{eq:constraints}. 
\begin{lemma}
\label{lemma:constraint_tightening}
Suppose the prediction~\eqref{eq:full_prediction} satisfies
\begin{align}
\label{eq:tightend_constraint}
g_j(z_{\mathrm{r}}(t))+L_{\mathrm{g},j}\delta_{\mathrm{z}}(t)\leq 0,~j=0,\dots,n_{\mathrm{c}}.
\end{align}
Then, $z(t)\in\mathcal{Z}$.
\end{lemma}
\begin{proof}
Given~\eqref{eq:robust_prediction} and \eqref{eq:tightend_constraint}, we have 
\begin{align*}
g_j(z(t))
\leq  g_j(z_{\mathrm{r}}(t))+L_{\mathrm{g},j} \delta_{\mathrm{z}}(t)\leq 0,~
j=0,\dots,n_{\mathrm{c}},
\end{align*}
i.e., $z(t)\in\mathcal{Z}$ with \eqref{eq:constraint_con} and Lipschitz constants $L_{\mathrm{g},j}\geq 0$. 
\end{proof}
The implementation with a Lipschitz constant $L_{\mathrm{g},j}$ is very natural for linear or conic inequality constraints. 
For more general constraints, Condition~(23) can  be replaced by
\begin{align*}
z_{\mathrm{r}}\in\bar{\mathcal{Z}}(\delta_{\mathrm{z}})
:=  \mathcal{Z}\ominus\mathcal{B}_\infty(\delta_{\mathrm{z}}),
\end{align*}
where $\mathcal{B}_{\infty}(\delta) := \{ e \mid \|e\|_{\infty} \le \delta \}$ denotes the infinity-norm ball and $\ominus$ the Pontryagin difference.

\subsubsection*{Reduced-order optimal control problem}
Given the robust prediction bound (Lemma~\ref{lemma:full_error_bound}) and the robust reformulation of the output constraints (Lemma~\ref{lemma:constraint_tightening}), we are ready to present the reduced-order 
formulation of the full-order optimal control problem~\eqref{eq:OCP}:
\begin{subequations}
\label{eq:OCP_ROM} 
\begin{align} 
\label{eq:OCP_ROM_cost}
\min_{u(\cdot)}&\int_0^T \|z_{\mathrm{r}}(t)-z_{\mathrm{ref}}(t)\|_{\infty}+R\cdot \|u(t)\|_\infty+ \delta_{\mathrm{z}}(t)~\mathrm{d}t\\
\label{eq:OCP_ROM_constraints}
\text{s.t. }&\eqref{eq:full_prediction},\eqref{eq:tightend_constraint},~u(t)\in\mathcal{U},~\forall t\in[0,T].
\end{align}
\end{subequations}
We assume that Problem~\eqref{eq:OCP_ROM} admits a feasible solution and we denote an optimal trajectory with a superscript $^\star$ and the minimal cost by $\mathcal{J}^\star$. 
Compared to Problem~\eqref{eq:OCP}, Problem~\eqref{eq:OCP_ROM} does not require simulating the large full-order state $x(t)\in\mathbb{R}^{\nxfull}$. 
In addition, it does not require joint optimization of the worst-case disturbances $w(t)$. 
Instead, we only need to simulate an ordinary differential equation of dimension $\nxreduced+\npsi+1$, where the order of $\nxreduced$ and $\npsi$ can be freely chosen. This results in a drastic reduction in computational complexity; especially as computational complexity scales cubically in the state dimension for long horizons~\cite[Sec.~3.3]{frison2020hpipm}. 
Compared to simply using the ROM~\eqref{eq:sys_ROM}, we additionally predict a bound $\delta_{\mathrm{z}}(t)$ on the prediction error using an error-bounding system~\eqref{eq:full_prediction_delta_dyn} with a dynamical filter~\eqref{eq:full_prediction_filter_dyn}, and we use this bound to tighten the constraint in~\eqref{eq:tightend_constraint} and upper-bound the cost in~\eqref{eq:OCP_ROM_cost}. 
The following theorem shows that the reduced-order optimal control problem~\eqref{eq:OCP_ROM} provides a feasible solution to the (computationally intractable) full-order optimal control problem~\eqref{eq:OCP}. 
\begin{theorem}
\label{thm:OCP}
Suppose Problem~\eqref{eq:OCP_ROM} is feasible.
Then, the returned input trajectory $u^\star(t)$ is a feasible solution to Problem~\eqref{eq:OCP}, i.e., applying this input to the full-order system~\eqref{eq:sys} yields a trajectory that satisfies the constraints~\eqref{eq:constraints} $\forall t\in[0,T]$ for any disturbances $w(t)\in\mathcal{W}$. 
Furthermore, the resulting cost~\eqref{eq:OCP_cost} is upper bounded by $\mathcal{J}^\star$. 
\end{theorem}
\begin{proof}
The fact that $u^\star(t)$ is a feasible solution to Problem~\eqref{eq:OCP} and ensures satisfaction of the output constraints follows from Lemma~\ref{lemma:constraint_tightening}. 
The bound on the optimal cost follows from Lemma~\ref{lemma:full_error_bound} with
\begin{align*}
\|z_{\mathrm{r}}(t)-z(t)\|_\infty\leq& \|z_{\mathrm{r}}(t)-z(t)\|\leq \delta_{\mathrm{z}}(t),\\
\|z(t)-z_{\mathrm{ref}}(t)\|_\infty\leq& \|z_{\mathrm{r}}(t)-z_{\mathrm{ref}}(t)\|_\infty+\delta_{\mathrm{z}}(t).
\qedhere
\end{align*}
\end{proof}
\begin{remark}(Receding horizon implementation)
The MPC formulation in Problem~\eqref{eq:OCP_ROM} is conceptually similar to the MOR-based MPC in~\cite{loehning2014model,alora2023robust}, with a ROM and an error-bounding system.
By leveraging this similarity, it is possible to show closed-loop properties for a receding-horizon implementation under small modifications. 
In particular, this requires a terminal set constraint $(x_{\mathrm{r}}(T),\bar{\psi}(T),\delta_\chi(T))\in\mathcal{X}_{\mathrm{f}}$, with an offline designed positive invariant set $\mathcal{X}_{\mathrm{f}}\subseteq\mathbb{R}^{\nxreduced+\npsi+1}$. 
Furthermore, the initial state constraint needs to be replaced by a
sufficient condition for $V_\chi(t)\leq \delta_\chi(t)$.
In case $x(t),\psi(t)$ can be measured, this can be directly implemented. 
For the more general case, an observer may be needed~\cite{schwenkel2025output}.  
Given these modifications, the MPC formulation can be used in a receding-horizon fashion and inherits all the theoretical guarantees from Theorem~\ref{thm:OCP}.
\end{remark}
\begin{remark}(Closed-loop prediction)
\label{rk:feedback}
The assumption of open-loop stability and open-loop inputs can be relaxed using a linear feedback $u(t)=Kx(t)+v(t)$~\cite{schwenkel2022model,schwenkel2025output}. 
The matrix $K$ can be jointly optimized offline to minimize the peak-to-peak gain in Proposition~\ref{prop:peak_gain} using the algorithm in~\cite{schwenkel2025multi}.  
In this case, the input cost and input constraints in Problem~\eqref{eq:OCP_ROM} need to include a term $\delta_{\mathrm{u}}$ that accounts for the maximal deviation between the nominal planned input and the actually applied input, similar to $\delta_{\mathrm{z}}$.  
\end{remark}

\section{Discussion \& qualitative comparison}
\label{sec:discussion}
In the following, we discuss qualitative properties of the proposed method and existing approaches. 

\subsection{Proposed approach}
\subsubsection*{Properties of proposed approach}
First, we highlight three main properties of the proposed methodology. 
\begin{enumerate}[label=\alph*)]
\item We provide rigorous guarantees on performance and constraint satisfaction for large-scale systems; 
\label{property_1}
\item The computational complexity of the Problem~\eqref{eq:OCP_ROM} depends on the user-chosen dimension of the ROM $\nxreduced$ instead of the full-order system $\nxfull\gg \nxreduced$;
\label{property_2}
\item The robust prediction error bound $\delta_{\mathrm{z}}$ depends on the filtered optimized trajectory $\bar{\psi}$ and enters both the cost and the tightened constraints in Problem~\eqref{eq:OCP_ROM}. Hence, the optimizer intrinsically favours input trajectories that lead to smaller bounds $\delta_{\mathrm{z}}$ on the prediction error.
\label{property_3}
\end{enumerate}
\subsubsection*{Computational demand}
Tailored MPC solvers scale linearly with the prediction horizon and cubic in the state and input dimension (cf. \cite[Sec.~3.3]{frison2020hpipm}). 
The proposed approach~\eqref{eq:OCP_ROM} can be implemented with the state $(x_{\mathrm{r}},\psi,\delta_\chi)\in\mathbb{R}^{\nxreduced+\npsi+1}$ and the optimized input given by the control $u\in\mathbb{R}^{\ninput}$ and one auxiliary variable to bound  $\|\bar{r}_\psi\|^2$ efficiently.  
Hence, the original full-order problem~\eqref{eq:OCP} is expected to be  $$\left(\dfrac{\nxfull+\ninput}{\nxreduced+\npsi+1+\ninput+1}\right)^3$$ times  computationally more expensive to solve. Even for the medium-scale problem in Section~\ref{sec:num}, this corresponds to a speed-up of a factor of over $100$ using the proposed approach.\footnote{%
In addition, reduced-order models tend to have reduced stiffness, which further enhances computational tractability~\cite{muntwiler2024stiffness}. However, Problem~\eqref{eq:OCP_ROM} is no longer a quadratic program, which also increases complexity.}
\subsubsection*{Filter design}
The addition of the dynamical filter in Section~\ref{sec:filter} further generalizes the dependence of the error bound on the optimized trajectory~\ref{property_3}. Specifically, the approach in Section~\ref{sec:peak} determines a bound that is proportional to the signal normal of $r$. 
The filter changes this to a norm on a filtered version of this signal, which provides more flexibility. 
In particular, for most reduced-order models, the model error depends largely on high-frequency excitation of unmodelled dynamics, while slow changes are well captured by the reduced-order model. This can be encoded by using a high-pass filter $\Psi$, where the cut-of frequency can be decided based on the unmodelled dynamics. 
In the numerical example (Sec.~\ref{sec:num}), we see that such a high-pass filter can capture the fact that the ROM is more accurate if the input $r(t)$ is smoother. Nonetheless, a more automatic choice of the filter $\Psi$ is part of current research.


\subsection{Qualitative comparison}
Next, we provide a qualitative comparison to existing approaches.
\subsubsection*{Uniform error bounds}
The majority of existing ROM-based MPC formulations~\cite{sopasakis2013constrained,kogel2015robust,lorenzetti2022linear,wang2022tube} achieve properties~\ref{property_1}--\ref{property_2}, but not \ref{property_3}. 
Specifically, these approaches typically impose compact constraints to limit the operation range of the ROM and then derive uniform bounds on the model mismatch. 
For the numerical comparison, we implement an approach based on the methodology in~\cite{lorenzetti2022linear}.
This approach derives uniform polyotpic bounds on $r(t)$ for  discrete-time dynamics and computes the exact worst-case error $z_{\mathrm{e}}(t)$ for a finite-horizon with a linear program. 
For the polytopic bounds on $r(t)$, we consider $\mathcal{W},\mathcal{U}$ polytopic and, as suggested in~\cite[Sec.~VI.C.1)]{lorenzetti2022linear}, we compute a hyperbox for $x_{\mathrm{r}}(t)$ using another linear program.

\subsubsection*{Error-bounding systems for MOR-based MPC}
Error-bounding systems similar to~\eqref{eq:error_bounding_system} have also been derived and leveraged for ROMs in~\cite{loehning2014model,alora2023robust,wirtz2014posteriori}. 
The key difference is how the constants, in particular constants comparable to $\lambda>0$ in~\eqref{eq:error_bounding_system}, are derived. 
In~\cite{wirtz2014posteriori}, $\lambda$ is given by the one-sided Lipschitz constant of $A$. 
In~\cite{alora2023robust}, a quantity comparable to $\lambda$ is chosen as the slowest eigenvalue of $A$. 
In~\cite{loehning2014model}, it is discussed that the choice of $\lambda$ is a trade-off and can be evaluated by using a Lyapunov equation, but no systematic design is suggested. 
These approaches do not provide a systematic fashion to minimize the error bounds (Sec.~\ref{sec:peak}). 
In addition, none of these approaches can leverage filtered signals (Sec.~\ref{sec:filter}).  
We will see in the numerical comparison that both differences have a significant impact on conservatism. 
In the numerical comparison, we implement the approach in~\cite{loehning2014model} by choosing $\lambda$ close to the slowest eigenvalue. 

\subsubsection*{IQC-based analysis}
ROMs yield dynamic uncertainty $\Delta$, which can be analysed using IQCs. 
Peak-to-peak analysis using IQCs has been derived in~\cite{schwenkel2023robust,schwenkel2025multi} and utilized for MPC in~\cite{schwenkel2022model,schwenkel2025output}. 
These results also yield a scalar error-bounding system similar to Proposition~\ref{prop:peak_gain}. 
IQCs characterize the dynamic uncertainty using filters, however, they do \emph{not} leverage filtered inputs in the way we proposed in Section~\ref{sec:filter}\ifbool{arxiv}{. Appendix~\ref{app:IQC} adapts existing 
IQC results~\cite{schwenkel2023robust} to the considered continuous-time case and provides more discussion.}{, see also~\cite[App.~A]{kohler2025ROM} for more discussion.}

\section{Numerical comparisons}
\label{sec:num} 
We first describe the considered problem and highlight issues in na\"ively applying a ROM-based MPC. 
Then we showcase the performance of the proposed ROM-based MPC and provide comparisons to existing methods. 
Implementation is done in Matlab on a standard Laptop. 
Code is open-source available:\\
{\small\url{https://github.com/KohlerJohannes/ROM_MPC_ECC}}
\vspace{1mm}
\subsubsection*{Setup}
We consider a chain of $N=50$ mass-spring-damper systems, where the $u(t)\in\mathbb{R}$ is the force on the last mass, $z(t)\in\mathbb{R}$ is the position of the first mass, see Figure~\ref{fig:illustration}. 
Each element has mass $1$, spring-constant $10$, damping-constant $20$, resulting in under-damped system of dimension $\nxfull=100$. 
We consider no disturbances, i.e., $w(t)=0$, and focus on the error resulting from ROMs.
\begin{figure}[t]
\includegraphics[width=0.48\textwidth]{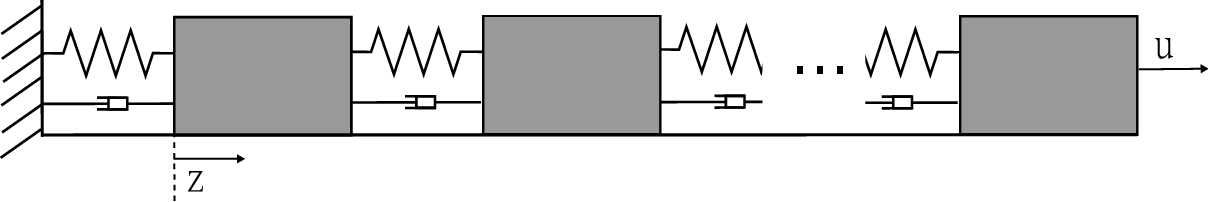}
\caption{Numerical example: chain of $N=50$ masses connected by spring and damper elements with first position $z$ and control input $u$.}
\label{fig:illustration}
\end{figure}

\subsubsection*{Na\"ive ROM-based MPC}
We use a ROM of dimension $\nxreduced=8$, given by a two dimensional lumped model and the slowest $6$ eigenmodes. 
First, we implement a na\"ive ROM-based MPC that neglects the mismatch.
The optimized trajectory and true full-order simulation can be seen in Figure~\ref{fig:ROM_nom}.
There is a rather large difference between the two trajectories, which also leads to constraint violations.  
\begin{figure}[t]
\includegraphics[width=0.48\textwidth]{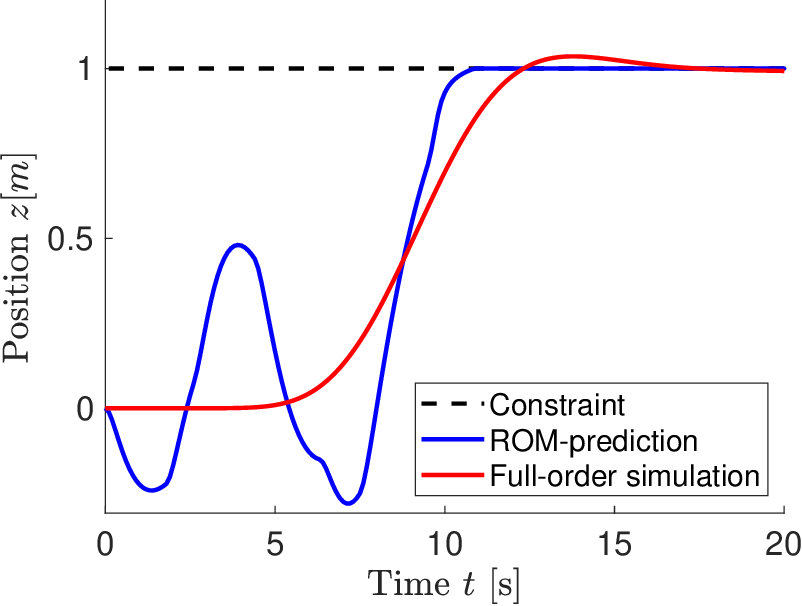}
\caption{Na\"ive ROM-based MPC: Trajectory optimized with nominal ROM (blue) shows significant deviation from full-order trajectory (red), and violates the constraints (black), since the structural model mismatch is not accounted for.}
\label{fig:ROM_nom}
\end{figure}

\subsubsection*{Proposed ROM-based MPC}
Next, we implement the proposed MPC formulation~\eqref{eq:OCP_ROM}.
The ROM is accurate for slow operation. 
We use this knowledge, by choosing the filter $\Psi$~\eqref{eq:filter} as a first-order high-pass with $n_\psi=\nxreduced+\ninput=9$.
In addition, the filter normalizes the inputs. Normalization and frequency were chosen heuristically.
Offline, we solve Problem~\eqref{eq:LMI_peak_filter}, where we consider $10$ equally spaced values of $\lambda_\chi>0$, which took $206$~[s].

The constraints are given by $\mathcal{U}=[-10^2,10^2]$ and $\mathcal{Z}=\{z|~z\leq 1\}$. 
The initial condition is $x_0=0$ and the desired reference is $z_{\mathrm{ref}}(t)=1$, i.e., the reference lies on the boundary of the constraint set. 
We use a horizon of $T=300~[s]$, an input penalty of $R=10^{-3}$. 
We implement the MPC~\eqref{eq:OCP_ROM}  by discretizing the linear dynamics with a sampling time of $2~[s]$, where we introduce additional decision variables to upper bound the quadratic term $\|r_{\psi}(t)\|^2$ in~\eqref{eq:full_prediction_filter_dyn}. 
Solving Problem~\eqref{eq:OCP_ROM} for a given initial state took about $0.7$~[s]. 

The optimized trajectory can be seen in Figure~\ref{fig:ROM_OCP}. 
The proposed approach provides a robust prediction of the full-order system and ensures constraint satisfaction and convergence to the reference. 
The mismatch between the ROM and the full-order model is significantly smaller compared to Figure~\ref{fig:ROM_nom}. 
This is an indirect consequence of Property~\ref{property_3}: the optimal solution 
yields small prediction errors since a robust bound on the prediction error $\delta_{\mathrm{z}}(t)$ is included in Problem~\eqref{eq:OCP_ROM}. 
The formulation also enables operation close to the constraints by optimizing inputs that yield vanishing small prediction errors. 
Such effects cannot be achieved when uniform error bounds are employed, which is standard in  the ROM-based MPC literature~\cite{sopasakis2013constrained,kogel2015robust,lorenzetti2022linear,wang2022tube}. 

\begin{figure}[t]
\includegraphics[width=0.48\textwidth]{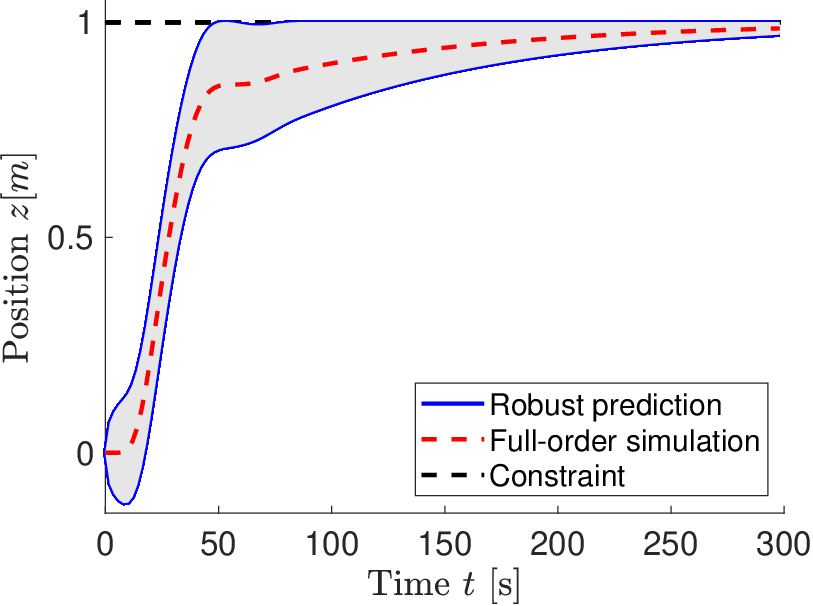}
\caption{Proposed robust ROM-based MPC~\eqref{eq:OCP_ROM}: Full-order simulation (red) contained in robust ROM-based prediction (grey shaded area with blue boundary), and below the constraint (black, dashed).}
\label{fig:ROM_OCP}
\end{figure}

\subsubsection*{Comparison}
We compare the conservatism of different robust prediction methods. 
We utilize the optimized trajectory seen in Figure~\ref{fig:ROM_OCP} and evaluate the size of the predicted reachable set around the ROM prediction for the following methods:
\begin{enumerate}[label=\alph*)]
\item \emph{Uniform:} Input-independent error bounds using compact constraints~\cite{lorenzetti2022linear};
\item  \emph{Input-dependent:} Error-bounding system from~\cite{loehning2014model};
\item  \emph{Peak}: Proposed error-bounding system (Sec.~\ref{sec:peak});
\item  \emph{Peak-filter}: Proposed error-bounding system using additional filters (Sec.~\ref{sec:filter}).
\end{enumerate}
The results are shown in Figure~\ref{fig:compare}. 
The proposed method reduces conservatism by approximately  four orders of magnitude compared to~\cite{lorenzetti2022linear} and \cite{loehning2014model}. 
When compared with the error-bounding system in~\cite{loehning2014model}, both minimizing the peak-to-peak gain and adding a suitable filter reduce conservatism by several orders of magnitude. 
Note that both \emph{Input-dependent}~\cite{loehning2014model} and \emph{Peak} (Sec.~\ref{sec:peak}) depend on the magnitude of the residual $\|r(t)\|$ and the large difference is primarily due the choice of the Lyapunov function used to sequentially propagate uncertainty.
In comparison to~\cite{lorenzetti2022linear}, the uniform bound is conservative, as a large model-mismatch is considered even during steady-state operation. Furthermore, in this example the worst-case mismatch is excessively large, making robust designs infeasible.  
Only \emph{Peak-filter} (Sec.~\ref{sec:filter}) yields a bound $\delta_{\mathrm{z}}(t)$ that decays to zero, since the high-passed filtered signal $r_\psi(t)$ is zero at steady-state operation.
\begin{figure}[t]
\includegraphics[width=0.48\textwidth]{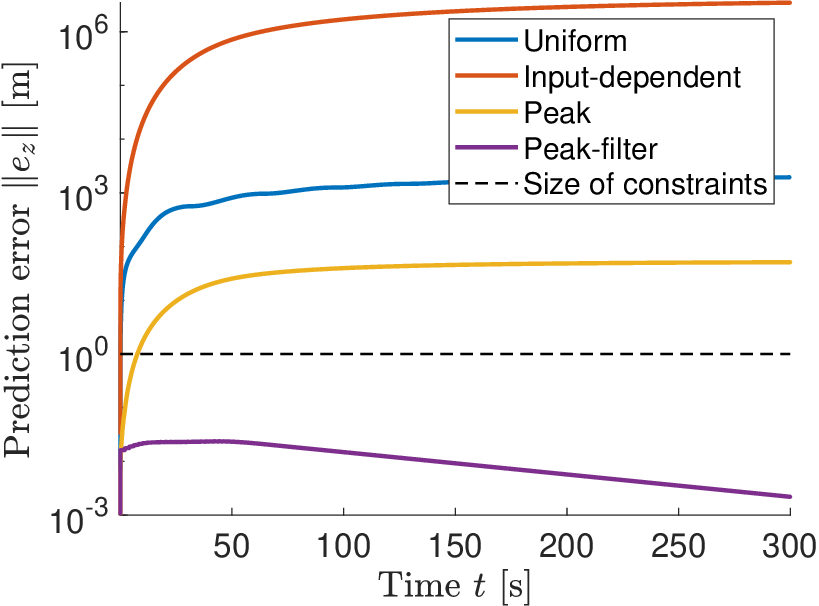}
\caption{Comparison of prediction error bounds of ROM. 
The dashed line indicates the size of the constraint set $\mathcal{Z}$, i.e., alternative methods provide bounds that are significantly larger than the constraint set.}
\label{fig:compare}
\end{figure}

\section{Conclusion}
We provided an MPC design for large-scale linear systems using a reduced-order model that enables efficient implementation while providing guarantees on constraint satisfaction and performance. 
Our key technical contribution is a scalar error-bounding system that provides reliable bounds on the prediction error of the reduced-order model based on a filtered predicted trajectory.
Numerical results show that the method can reduce conservatism by orders of magnitude.

Future work will focus on improving the offline design by systematically optimizing both the reduced-order model and the filter, while enhancing scalability to enable hardware experiments with soft robots.

\bibliographystyle{IEEEtran}
\bibliography{Literature}

\ifbool{arxiv}{%
\appendix
\subsection{IQC-based peak-to-peak reachability analysis}
\label{app:IQC}
In the following, we provide a general peak-to-peak gain analysis for robust reachability of linear uncertain systems characterised using IQCs. 
The following exposition is analogous to the discrete-time results in~\cite{schwenkel2023robust}. We first consider a general setup of uncertain dynamics characterized by IQCs, provide a reachability result for this general problem, and then discuss application to ROMs.
\subsubsection*{Setup}
We consider a linear uncertain system of the form
\begin{subequations} 
\label{eq:dynamics_IQC}
\begin{align}
\dot{\xi}(t)=&A_\xi \xi(t)+B_\xi w_\xi(t),~\xi(0)=\xi_0,\\
y_\xi(t)=&C_\xi \xi(t)+D_\xi w_\xi(t),\\
z_{\mathrm{e}}(t)=&C_{\mathrm{z}} \xi(t).
\end{align}
\end{subequations}
The state $\xi(t)\in\mathbb{R}^{n_\xi}$ typically encompasses the state of the prediction model and additional filters. 
The generalized disturbances $w_\xi(t)$ encompasses all external signals, including known control inputs, bounded disturbances, and error from dynamic uncertainty. 
Known behaviour of the external signals is captured through the output $y_{\xi}(t)$, by assuming that it satisfies an $\alpha$-IQC~\cite{hu2016exponential}: 
There exist known (indefinite) matrices $M,X$, such that 
\begin{align}
 \label{eq:IQC_def}
 \int_0^t e^{2\alpha \tau} y_\xi(\tau)^\top M y_\xi(\tau) \mathrm{d}\tau +e^{2\alpha t} \xi(t)^\top X \xi(t)\geq 0 \quad \forall t\geq 0. 
\end{align}
\subsubsection*{Robust peak-to-peak gain analysis}
\begin{proposition}
\label{prop:IQC}
Suppose the following optimization problem admits a feasible solution $P_\xi$, $\Gamma\succeq 0$, $\alpha\geq 0$:
\begin{subequations}
\label{eq:IQC_LMI}
\begin{align} 
\label{eq:IQC_LMI_alpha}
& \left[ \vphantom{
 \begin{bmatrix}
 x \\ x \\ x \\ x
 \end{bmatrix}
 } 
 \star \right]^{\top} 
 \left[
 \begin{array}{cc|cc}
 2 \alpha P_\xi & P_\xi & 0 & 0 \\
 P_\xi & 0 & 0 & 0 \\
 \hline
 0 & 0& M & 0 \\
 0& 0 & 0&-\Gamma \\
 \end{array}
 \right] 
 \begin{bmatrix}
 I_{n_\xi}& 0 \\
 A_\xi & B_\xi\\
 C_\xi& D_\xi \\
 0 & I_{n_{\mathrm{w}_\xi}}
\end{bmatrix}
\prec 0,\\
\label{eq:IQC_LMI_output}
& \begin{bmatrix}
\gamma I_{\nz}&C_{\mathrm{z}}\\
C_{\mathrm{z}}^\top & \alpha (P_\xi-X)
\end{bmatrix}\succeq 0,~P_\xi-X\succeq 0.
\end{align}
\end{subequations}
Then, for any signal $w_\xi(t)$ satisfying~\eqref{eq:IQC_def}, it holds that
\begin{align}
\label{eq:IQC_error_bound}
\dfrac{1}{\alpha\gamma}\|z_{\mathrm{e}}(t)\|^2\leq \xi(t)^\top (P_\xi-X)\xi(t) \leq \delta_\xi(t),~\forall t\geq 0,
\end{align}
with
\begin{align}
\label{eq:IQC_delta_dyn}
\dot{\delta}_\xi(t)=-2\alpha\delta_\xi(t)+ \|w_\xi(t)\|_\Gamma^2,~\delta_\xi(0)\geq \xi_0^\top (P_\xi-X)\xi_0. 
\end{align}
\end{proposition}
\begin{proof}
Denote $V_\xi(t)=\xi(t)^\top P_\xi \xi(t)$. 
Multiplying~\eqref{eq:IQC_LMI_alpha} from left and right by $[\xi(t)^\top,~w_\xi(t)^\top]$ and its transposed yields
\begin{align*}
2\alpha V_\xi(t)+\dot{V}_\xi(t)+y_\xi(t)^\top M y_\xi(t)-\|w_\xi(t)\|_\Gamma^2\leq 0.
\end{align*}
Multiplying the expression by $\exp(2\alpha t)$ and integrating from $0$ to $t$ yields
\begin{align*}
&e^{2\alpha t}V_\xi(t)-V_\xi(0)\\
\leq& \int_0^t e^{2\alpha \tau} (\|w_\xi(\tau)\|_{\Gamma}^2-y_\xi(\tau)^\top M y_\xi(\tau)) \mathrm{d}\tau\\
\stackrel{\eqref{eq:IQC_def}}{\leq} & \int_0^t e^{2\alpha \tau}\|w_\xi(\tau)\|_{\Gamma}^2 \mathrm{d}\tau+ e^{2\alpha t}\xi(t)^\top X \xi(t). \end{align*}
Integrating the dynamics~\eqref{eq:IQC_delta_dyn} yields
\begin{align}
\delta_\xi(t)\leq e^{-2\alpha t}\delta_\xi(0)+\int_0^te ^{-2\alpha(t-\tau)}\|w_\xi(\tau)\|_\Gamma^2 \mathrm{d}\tau.
\end{align}
Comparing the two expressions, yields the second inequality in~\eqref{eq:IQC_error_bound}. 
Applying a Schur complement to~\eqref{eq:IQC_LMI_output} yields
\begin{align*}
-\alpha (P_\xi-X)+\dfrac{1}{\gamma}C_{\mathrm{z}}^\top C_{\mathrm{z}}\preceq 0.
\end{align*}
Multiplying from left and right by $\xi(t)$ and its transposed yields the first inequality in~\eqref{eq:IQC_error_bound}.
\end{proof}
\subsubsection*{Relation to robust prediction with ROMs}
In the following, we discuss how Proposition~\ref{prop:IQC} relates to the results in Propositions~\ref{prop:peak_gain}-\ref{prop:filter_peak_gain}. 
First, if we neglect the IQC~\eqref{eq:IQC_def} by setting $M=0$, $X=0$, we can recover Proposition~\ref{prop:peak_gain} by choosing $\xi$ as the error $e\in\mathbb{R}^{\nxfull}$, $w_\xi=r$, and $\Gamma=\gamma I_{n_{\mathrm{w}_\xi}}$. 
The IQC~\eqref{eq:IQC_def} can also provide a characterization of the unmodelled dynamics $\Delta$~\eqref{eq:error_dyn} using dynamical filters. 
Thus, Problem~\eqref{eq:IQC_LMI} may be lower dimensional, which would also facilitate the scalability of the offline design. 
In this case, semi-definiteness of  $\Gamma\succeq 0$ is important as many components of $w_\xi(t)$ cannot be predicted, but $\|w_\xi(t)\|_{\Gamma}$ is required to implement~\eqref{eq:IQC_delta_dyn}.  
Despite the fact that IQCs leverage dynamical filters, the filter-based approach derived in Section~\ref{sec:filter} is in fact \emph{not} a special case.
While IQCs use filters to described the model error, our approach uses filters to re-define the peak-to-peak norm. 
We except that a more natural unification of these two approaches is possible, which is, however, beyond the scope of this paper.}{}
\end{document}